\documentclass[conference]{IEEEtran}

\usepackage{graphicx}
\usepackage{amsmath}
\usepackage{enumitem}
\usepackage{bbm}

\usepackage{flushend}

\usepackage{multirow}

\usepackage{mathrsfs}
\usepackage{color,soul}
\usepackage{algorithm}
\usepackage{algpseudocode}
\usepackage{amsmath,amssymb}
\usepackage{amsmath,amsfonts,amssymb} 

\usepackage{geometry}
\usepackage{amsthm}

\usepackage{eurosym}
\usepackage{subfig}
\newtheorem{theorem}{Theorem}[section]

\title{{Optimal Storage Arbitrage under Net Metering using Linear Programming}}
\author{\IEEEauthorblockN{Md Umar Hashmi\IEEEauthorrefmark{1},
		Arpan Mukhopadhyay\IEEEauthorrefmark{2},
		Ana Bu\v{s}i\'c\IEEEauthorrefmark{1},
		Jocelyne Elias\IEEEauthorrefmark{3}, 
		and
		Diego Kiedanski\IEEEauthorrefmark{4}}
	\IEEEauthorblockA{\IEEEauthorrefmark{1} INRIA and the Computer Science Dept. of Ecole Normale Sup\'erieure, CNRS, PSL Research University,
		Paris, France}
	\IEEEauthorblockA{\IEEEauthorrefmark{2} the Department of Computer Science, University of Warwick, the UK}
	\IEEEauthorblockA{\IEEEauthorrefmark{3} Laboratoire d'Informatique PAris Descartes (LIPADE), Universit\'e Paris Descartes, Paris, France}
	\IEEEauthorblockA{\IEEEauthorrefmark{4}T\'el\'ecom Paristech, 23 Avenue d'Italie, Paris, France}
}

\IEEEoverridecommandlockouts
\IEEEpubid{\makebox[\columnwidth]{978-1-5386-8099-5/19/\$31.00~\copyright2019 IEEE \hfill} \hspace{\columnsep}\makebox[\columnwidth]{ }}

\begin{document}
	\maketitle
	\begin{abstract}
		We formulate the optimal energy arbitrage problem for a piecewise linear cost function for energy storage devices using linear programming (LP). 
		The LP formulation is based on the equivalent minimization of the epigraph.
		This formulation considers ramping and capacity constraints, charging and discharging efficiency losses of the storage, inelastic consumer load and local renewable generation in presence of net-metering which facilitates selling of energy to the grid and incentivizes consumers to install renewable generation and energy storage.
		We consider the case where the consumer loads, electricity prices, and renewable generations at different instances are uncertain. These uncertain quantities are predicted using an Auto-Regressive Moving Average (ARMA) model and used in a model predictive control (MPC) framework to obtain the arbitrage decision at each instance.  
		In numerical results we present the sensitivity analysis of storage performing arbitrage with varying ramping batteries and different ratio of selling and buying price of electricity. 
	\end{abstract}
	
	\vspace{2mm}
	\begin{IEEEkeywords}
		Energy arbitrage, Battery, Linear programming, Net-metering, Model Predictive Control
	\end{IEEEkeywords}
	
\section{Introduction}
%
Energy storage devices provide flexibility to alter the consumption behavior of an electricity consumer. 
Storage owners at the consumer side could participate in demand response, energy arbitrage, peak demand shaving, power backup to name a few \cite{xi2014stochastic}, \cite{hashmi2019energy}.
These features
of storage devices will be more lucrative for storage owners
with the growth of intermittent generation sources which
increase volatility on the generation side in power network \cite{hashmi2018effect}.
%
%
Furthermore, batteries are becoming more affordable making several applications of storage devices financially viable.
%
Storage devices can perform arbitrage of energy with time varying consumer load, distributed generation production and electricity price. Furthermore, utilities promote inclusion of distributed generation and storage deployment by introducing net-metering. Net energy metering (NEM) or net-metering refers to the rate consumers receive for feeding power back to the grid. Most NEM policies indicate that consumers receive a rate at best equal to the buying price of electricity \cite{wiki_net}.
%
Authors in \cite{hashmi2017optimal} consider storage operation under equal buy and sell price case. This framework is generalized in \cite{hashmi2018netmetering}, covering cases where the ratio of buy and sell price could arbitrarily vary between 0 and 1.
For equal buying and selling price, the storage control becomes independent of inelastic load and renewable generation of the consumer \cite{hashmi2017optimal}, \cite{xu2017optimal}.
 The cost function considered in this work includes inelastic load, renewable generation and storage charging and discharging efficiency, and ramping and capacity constraints.
We first show that the cost function, based on the selection of the optimization variable, is convex and piecewise linear. Then, we formulate the optimal arbitrage problem for an electricity consumer with renewable generation adopting NEM by using Linear Programming (LP).

Authors in \cite{zidar2016review} provide a summary of storage control methodologies used in power distribution networks among which LP based formulations can be solved efficiently using commercially available solvers.
Therefore, these algorithms can be used to efficiently solve the arbitrage problem for the duration of a day  divided into smaller time steps ranging from 5 minutes to an hour.
A day is the typical time horizon over which arbitrage is performed \cite{mokrian2006stochastic,hu2010optimal}.

{Authors in \cite{cruise2019control} formulate the optimal arbitrage problem for a strictly convex cost function and observe that for a piecewise linear convex cost function, as in \cite{hashmi2017optimal}, LP-based tools can be applied. }
LP techniques for energy storage arbitrage have been used in several prior works: \cite{park2017linear}, \cite{byrne2015potential}, \cite{chouhan2016optimization}, \cite{thatte2013risk}, \cite{bradbury2014economic}, \cite{nguyen2018maximizing}, \cite{wang2018energy}.
Authors in \cite{bradbury2014economic, byrne2015potential, wang2018energy} consider storage operation in presence of time-varying electricity price.
However, in these formulations no renewable energy source or consumer load is assumed to be present.
 Authors in \cite{chouhan2016optimization, thatte2013risk} consider optimal scheduling of storage battery for maximizing energy arbitrage revenue in presence of distributed energy resources and variable electricity price. 
Formulations presented in \cite{nguyen2018maximizing, park2017linear} consider storage performing arbitrage in a residential setting with inelastic load and local generation.
Most common
LP formulations for energy arbitrage such as in \cite{park2017linear}, \cite{wang2018energy}, \cite{thatte2013risk}, \cite{byrne2015potential} consider separation of charging and discharging
components {of storage ramping variable}.
In these formulations, they do not include constraint
enforcing only one of the charging or the discharging component to be active at any particular time as the inclusion of such a constraint makes these formulations nonlinear. 
{In the absence of such a constraint charging and discharging component can compensate each other which can lead to suboptimal solution.
	 }
Authors in \cite{chouhan2016optimization, bradbury2014economic} do not consider energy storage charging and discharging efficiencies in the cost minimization, making it straightforward to apply LP.
Authors in \cite{nguyen2018maximizing} consider a special case of optimization with zero-sum aggregate storage power output. For such a case LP tools could be used, however, generalizing the formulations needs to be explored further.


The key contributions of this paper are as follows:\\
$\quad \bullet$ \textit{LP formulation for storage control}: 
%
We formulate the LP optimization problem for piecewise linear convex cost function, for storage with efficiency losses, ramping and capacity constraints and a consumer with inelastic load and renewable generation. The buying and selling price of electricity are varying over time. The selling price is assumed to be at best equal to buying price for each time instant, this assumption is in sync with most net-metering policies worldwide.
%
%
Based on the convex and piecewise linear structure of the cost function we apply an epigraph based minimization described in \cite{boyd2004convex} to the arbitrage problem.
The reduction of this formulation for 
 (a) lossless battery with equal buying and selling price of electricity and (b) lossy battery with selling price less than or equal to buying price, is trivial and not included in this paper.  \\
$\quad \bullet$ \textit{Real-time implementation:} We implement an auto-regressive based forecast model along with model predictive control and 
numerically analyze their effect on arbitrage gains using 
real data from a household in Madeira in Portugal and electricity price from California ISO \cite{ENOnline}. The effect of parameter uncertainty on arbitrage gains is more pronounced for cases where selling price is comparable to buying price.\\
$\quad \bullet$ \textit{Sensitivity of ratio of selling and buying price}: We numerically analyze the effect of the ratio of buying and selling price of electricity on the value of storage with inelastic load and renewable generation. We observe that the value of storage performing arbitrage significantly increases in the presence of load and renewable generation with the increasing {difference} of selling and buying price of electricity, compared to only storage performing arbitrage.
Inclusion of storage in the presence of load and renewable generation can be profitable even for cases where the selling price is zero or small compared to buying price. For the same case, only storage performing arbitrage would not be profitable.

%

The paper is organized as follows. Section~\ref{lpsec2} provides the description of the system. 
Section~\ref{lpsec3} presents the LP formulation of storage performing arbitrage with inelastic load, renewable generation and net-metering based compensation.
Section~\ref{lpsec4} presents an online algorithm using the proposed optimal arbitrage algorithm along with auto-regressive forecasting in the  MPC framework. 
Section~\ref{lpsec5} discusses numerical results. 
Finally, Section~\ref{lpsec6} concludes the paper.

\section{System Description}
\label{lpsec2}
We consider a consumer of electricity over a fixed period of time. 
The consumer is assumed to be equipped with a rooftop solar photovoltaic and a battery to store excess generation. It is also connected to the electricity grid
from where it can buy or to which it can sell energy.
The total duration, $T$, 
of operation is divided into $N$ steps indexed by $\{1,...,N\}$. 
The duration of step $i \in \{1,...,N\}$ is denoted as $h_i$. Hence, $T=\sum_{i=1}^{N} h_i$.
The price of electricity, $p_{\text{elec}}(i)$, equals the buying price, $p_b(i)$, if the consumption is positive; otherwise $p_{\text{elec}}(i)$ equals the selling price, $p_s(i)$; denoted as \vspace{-5pt}
\begin{equation}
p_{\text{elec}}(i)=
\begin{cases}
p_b(i) ,& \text{if consumption } \geq 0 ,\\
p_s(i) , & \text{otherwise,} \vspace{-5pt}
\end{cases}
\end{equation} 
Note $p_{\text{elec}}$ is ex-ante and the consumer is a price taker.
The ratio of selling and buying price at time $i$ is denoted as  \vspace{-5pt}
\begin{equation}
\kappa_i = {p_s(i)}/{p_b(i)}.  \vspace{-7pt}
\end{equation}
The end user inelastic consumption in time step $i$  is denoted as $d_i$ and
renewable generation as $r_i$.
Net energy consumption without storage is denoted as
$
z_i = d_i - r_i ~ \in \mathbb{R}.
$
Fig.~\ref{systemblock} shows the block diagram of the system, i.e., an electricity consumer with renewable generation and battery.
The efficiency of charging and discharging of the 
battery are denoted by $\eta_{\text{ch}}, \eta_{\text{dis}} \in (0,1]$, respectively. 
We denote the change in the energy level of the battery at $i^{\text{th}}$ instant by $x_i= h_i  \delta_i$,
where $\delta_i$ denotes the storage ramp rate at $i^{\text{th}}$ instant such that $\delta_i \in [\delta_{\min}, \delta_{max}]$ $\forall i$ and $\delta_{\min} \leq 0,\delta_{\max} \geq 0$ are the minimum and the maximum ramp rates (kW); 
$\delta_i > 0$ implies charging and $\delta_i < 0$ implies discharging. 
Energy consumed by the storage in the $i^{th}$ instant is given by  \vspace{-3pt}
\begin{equation}
s_i =f(x_i)= \frac{1}{\eta_{\text{ch}}}[x_i]^+ - \eta_{\text{dis}}[x_i]^-, \vspace{-3pt}
\end{equation}
%
where $x_i$ must lie in the range  from $X_{\min}^i=\delta_{\min}h_i$ to $X_{\max}^i= {\delta_{\max}h_i}$.
Note $[x_i]^+ = \max(0, x_i)$ and $[x_i]^- = \max(0, -x_i)$.
Alternatively, we can write $x_i =  \eta_{\text{ch}}[s_i]^+ - \frac{1}{\eta_{\text{dis}}}[s_i]^-$.
The limits on $s_i$  are given as $ s_i \in [S_{\min}^i, S_{\max}^i]$, where $S_{\min}^i=\eta_{\text{dis}}\delta_{\min}h_i$ and $S_{\max}^i= \frac{\delta_{\max}h_i}{\eta_{\text{ch}}}$.

Let $b_i$ denote the energy stored in the battery at the $i^{\textrm{th}}$ step. 
The battery capacity is defined as \vspace{-3pt}
\begin{equation}
b_i = b_{i-1} + x_i, \quad b_i\in [b_{\min},b_{\max}],  \forall i, \vspace{-3pt}
\end{equation}
where $b_{\min}, b_{\max}$ are the minimum and the maximum battery capacity. 
The total energy consumed between time step $i$ and $i+1$ is given as $L_i = z_i+s_i$. \vspace{-9pt}
\begin{figure}[!htbp]
	\center
	\includegraphics[width=2.9in]{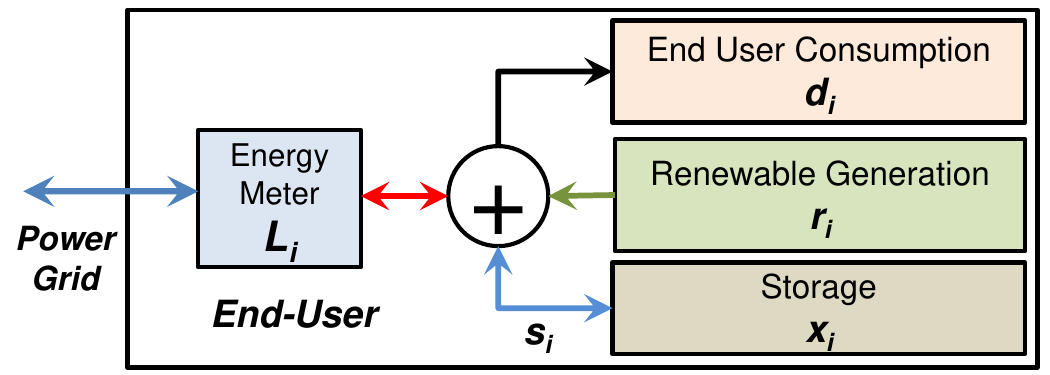} \vspace{-6pt}
	\caption{\small{Behind-the-meter electricity consumer with inelastic consumption, renewable generation and energy storage.}}\label{systemblock}
\end{figure}
\vspace{-9pt}

The battery operational life is often quantified using cycle and calendar life which decides the cycles a battery should perform over a time period.  
Friction coefficient, denoted as $\eta_{\text{fric}} \in [0,1]$, and introduced in \cite{hashmi2018limiting} assists in reducing the operational life of the battery such that low returning transactions of charging and discharging are eliminated, thus increasing the operational life of the battery.
In subsequent work, authors in \cite{hashmi2018long, hashmi2018pfcpowertech} propose a framework to tune the value of friction coefficient for increasing operational life of battery by eliminating low returning transactions.
\subsection{Arbitrage under Net-Metering}
The optimal arbitrage problem is defined as the minimization of the 
cost of total energy consumption subject to the battery constraints.
It is given as follows:
\begin{gather*}
\text{($P_{\text{NEM}}$)  } 
\min  \sum_{i=1}^N C_{nm}^{i}(x_i),
\end{gather*}
 \vspace{-9pt}
 \begin{gather*}
\vspace{2pt}\text{subject to, } 
b_{\min} - b_0\leq \sum_{j=1}^i x_j \leq b_{\max}- b_0 , \forall i \in \{1,..,N\},
\end{gather*}
\vspace{-9pt}
\begin{gather*}
x_i  \in \left[X_{\min}^i , X_{\max}^i\right]  \forall i \in \{1,..,N\},
\end{gather*}
where $C_{\text{nm}}^{i}(x_i)$ denotes the energy consumption cost function at instant $i$ and is given by
\begin{equation}
C_{\text{nm}}^{i}(x_i) = [z_i + f(x_i)]^+ p_b(i) - [z_i + f(x_i)]^- p_s(i).
\end{equation}

{Now we will show that} the optimal arbitrage problem is convex in $x=(x_i, i=1:N)$.
For this convexity to hold we require $p_b(i) \geq p_s(i)$ for all $i=1:N$, i.e., {$\kappa_i \in [0,1]$}.
The proposed framework is applicable for the case where selling price of electricity for the end user is lower than the buying price. This assumption is quite realistic as this is generally the case in most practical net metering policies \cite{wiki_net}. 
\begin{theorem}
	\label{thm:convexity}
	If $p_b(i) \geq p_s(i)$ for all $i=1:N$, then problem ($P_{\text{NEM}}$) is convex in $x$. 
\end{theorem} 
	\begin{proof}
		Let $\psi(t) =a[t]^+ - b[t]^-$ with $a\geq b \geq 0$. Using $t=[t]^+- [t]^-$ we have $\psi(t) = (a-b)[t]^+ + bt$. 
		Since both $[t]^+$ and $t$ are convex in $t$ and $a-b, b\geq 0$ we have that $\psi$ is convex since it is the positive sum of two convex functions. 
		
		Now let $f(x)= \frac{1}{\eta_{ch}} [x]^+ - \eta_{dis}[x]^-$ and $G_i(s) =[z_i+s]^+p_b(i) -[z_i+s]^-p_s(i)$.
		Then by the above reasoning we have that for $p_b(i) \geq p_s(i) \geq 0$ and $\eta_{ch}, \eta_{dis} \in (0,1]$,
		$G_i$ is convex in $s$ and $f$ is convex in $x$. Also, note that $G_i$ is non-decreasing in $s$. Hence, for $\lambda \in [0,1]$ we have
		\begin{align}
		G_i\big(f(\lambda x +(1-\lambda)y)\big) &\leq G_i\big(\lambda f(x) + (1-\lambda)f(y)\big)\\
		&\leq \lambda G_i(f(x)) + (1-\lambda)G_i(f(y))
		\end{align}
		In the above, the first inequality follows from the convexity of $f$ and non-decreasing nature of $G_i$
		and the second inequality follows from convexity of $G_i$. Therefore, we have that $G_i\circ f=G_i(f())$ is 
		a convex function in $x$. This shows that the objective function of ($P_{\text{NEM}}$) is convex in $x$ since $C_{nm}^i=G_i\circ f$.
		Since the constraints are linear in $x$ thus problem ($P_{\text{NEM}}$) is convex.
	\end{proof}

\section{Optimal Arbitrage with Linear Programming}
\label{lpsec3}
The optimal arbitrage problem, ($P_{\text{NEM}}$), can be solved using linear programming as the cost function is (i) convex and (ii) piecewise linear, and (iii) the associated ramping and capacity constraints are linear. 
In this section, we provide an LP formulation for the optimal arbitrage of the storage device under net-metering and consumer inelastic load and renewable generation, leveraging the epigraph based minimization presented in \cite{boyd2004convex}. 
A summary of the epigraph based formulation for a piecewise linear convex cost function is presented in Appendix~\ref{epigraphsec}.
The optimal arbitrage formulation for storage under net-metering and consumer inelastic load and renewable generation using the epigraph formulation is presented in this section.
Fig.~\ref{costlp} shows the optimal arbitrage cost function depending on the net-load without storage output, i.e. for $z_i\geq0$ and $z_i<0$. 
{Notice that the cost function $C_{nm}(i)$ is formed of 4 unique segments whose slopes, x-intercept and y-intercept are shown in Fig.~\ref{costlp} and listed in Table~\ref{segproperties}.
Fig.~\ref{costlp} shows that the inactive segment of the cost function denoted in green lies below the cost function denoted in red. Due to the convexity property of the cost function, it can be denoted as} 
\begin{equation}
\begin{split}
C_{nm}(i) = \max\big(\text{Segment 1, Segment 2},\\ \text{Segment 3, Segment 4}\big).
\end{split}
\label{signindependent}
\end{equation}
\vspace{-10pt}
\begin{figure}[!htbp]
	\center
	\includegraphics[width=3.4in]{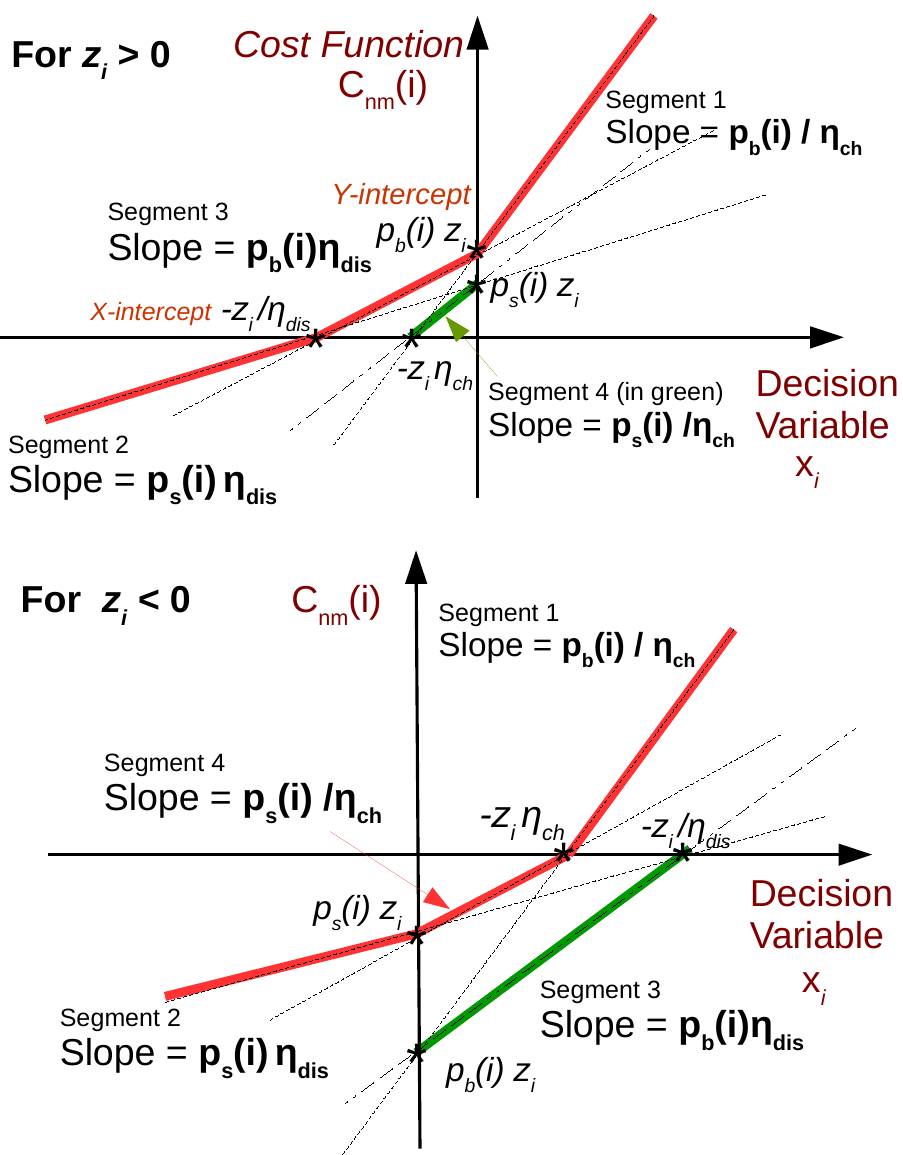} \vspace{-5pt}
	\caption{\small{The cost function segment wise for positive and negative net load $z$ \cite{hashmi2018netmetering}. The decision variable is storage change in charge level, $x_i$, and cost function, $C_{nm}(i)$ is formed with 4 unique line segments.}}\label{costlp}
\end{figure}

The epigraph based LP formulation is possible as irrespective of the sign of the load ($z$) the cost function can be represented as the maximum of the segments, shown in Eq.\ref{signindependent}.
Using the epigraph equivalent formulation for piecewise linear convex cost function we formulate the optimal arbitrage problem using linear programming, denoted as $\text{P}_{\text{LP}}$
\begin{gather*}
(\text{P}_{\text{LP}})~~\min \quad \{t_1 + t_2+...+t_N\}, \\
\text{subject to, }~~
\text{(a) Segment 1:~}\frac{p_b^i}{\eta_{ch}} x_i + z_i p_b^i  \leq t_i, ~\forall i \\
\text{(b) Segment 2:~}  {p_s^i}{\eta_{dis}} x_i + z_i p_s^i \leq  t_i, ~\forall i\\
\text{(c) Segment 3:~} {p_b^i}{\eta_{dis}} x_i + z_i p_b^i \leq  t_i, ~\forall i\\
\text{(d) Segment 4:~} \frac{p_s^i}{\eta_{ch}} x_i + z_i p_s^i \leq  t_i, ~\forall i\\
\text{(e) Ramp constraint:~} x_i \in [X_{\min}^i, X_{\max}^i], ~\forall i\\
\text{(f) Capacity constraint:~} \sum {x_i} \in [b_{\min}-b_0, b_{\max}-b_0],~ \forall i.
\end{gather*}
\vspace{-8pt}
\begin{table}[!htbp]
	\caption {\small{Cost function for storage with load under NEM}}
	\vspace{-6pt}
	\label{segproperties}
	\begin{center}
		\begin{tabular}{| c | c| c|c|}
			\hline
			Segment& Slope &  x-intercept& y-intercept \\ 
			\hline
			Segment 1 & $p_b(i)/\eta_{ch}$ &$-z_i \eta_{ch}$ & $z_ip_b(i)$ \\
			Segment 2 & $p_s(i)\eta_{dis}$ &$-z_i/ \eta_{dis}$ & $z_ip_s(i)$ \\
			Segment 3 & $p_b(i)\eta_{dis}$ &$-z_i/ \eta_{dis}$ & $z_ip_b(i)$ \\
			Segment 4 & $p_s(i)/\eta_{ch}$ &$-z_i \eta_{ch}$ & $z_ip_s(i)$ \\
			\hline
		\end{tabular}
		\hfill\
	\end{center}
\end{table}


The cost function for only lossy storage operation under NEM has two-piecewise linear segments and it is linear for equal buying and selling price of electricity with lossless battery. 
Authors in \cite{chouhan2016optimization, bradbury2014economic} present this case in their LP formulation.
This case can be obtained by simplifying the more general case depicted as $\text{P}_{\text{LP}}$ for cost function presented in Fig.~\ref{costlp}.


The LP based optimal arbitrage code described in this paper are publicly available at \texttt{\footnotesize{github.com/umar-hashmi/linearprogrammingarbitrage}}.

\section{Real-time implementation}
\label{lpsec4}
The previous section discussed optimal storage arbitrage under complete knowledge of future net loads and prices. In this section, we consider the setting where future values may be unknown. To that end, we first develop a forecast model for net load without storage (which includes inelastic consumer load and consumer distributed generation) and electricity price for future times, where the forecast is updated after each time step.
Then, we develop the forecasting model for net load with solar generation using AutoRegressive Moving Average (ARMA) model and electricity price forecast using AutoRegressive Integrated Moving Average (ARIMA).

The forecast models based on ARMA and ARIMA model developed in \cite{hashmi2019arbitrage} are used in this work.
 The forecast values are fed to a Model Predictive Control (MPC) scheme to identity the optimal modes of operation of storage for the current time-instance. 
 These steps (forecast and MPC) are repeated sequentially and highlighted in online Algorithm~\ref{algLPuncertainty}: \texttt{ForecastMPClinearProgram}.


\begin{algorithm}
	\small{\textbf{Storage Parameters}: {$\eta_{\text{ch}}, \eta_{\text{dis}}, \delta_{\max}, \delta_{\min}, b_{\max}, b_{\min}$, $b_0$}}.\\
	\small{\textbf{Inputs}: {$h, N, T,i=0 $, Rolling horizon optimization time period $ N_{\text{opt}}$, ~ Historical inelastic load, renewable generation and electricity price data}}.	
	\begin{algorithmic}[1]
		\State Use historical data to tune ARMA and ARIMA models,
		\While{$i < N$}
		\State Increment $i=i+1$,
		\State Real-time electricity price value $p_{\text{elec}}(i)$ and load $z_i$,
		\State Forecast $\hat{z}$ from time step $i+1$ to $i+ N_{\text{opt}}$ using ARMA,
		\State Forecast $\hat{p}_{b}$ and $\hat{p}_{s}$ from time $i+1$ to $i+ N_{\text{opt}}$ using ARIMA,
		\State Calculate $\hat{\kappa}$ as the ratio of $\hat{p}_{s}$ and $\hat{p}_{b}$,
		\State Build LP matrices for time step $i$ to $N$,
		\State Solve the Linear Optimization problem for forecast vectors,
		\State Calculate ${b_i}^*= b_{i-1}+\hat{x}^*(1)$,
		\State Update $b_0={b_i}^*$, the initial capacity of battery is updated.
		\State Return ${b_i}^*$, ${x_i}^*$.
		\EndWhile
	\end{algorithmic}
	\caption{\texttt{ForecastMPClinearProgram}}\label{algLPuncertainty}
\end{algorithm}

\section{Numerical Results}
\label{lpsec5}
For the numerical evaluation, we use battery parameters listed in Table~\ref{parametersBatlp}.
The performance indices used for evaluating simulations are:\\ 
$\quad \bullet$  \textit{Arbitrage Gains:} denotes the gains (in absence of load and renewable) or reduction in the cost of consumption (made in presence of load and renewable) due to storage performing energy arbitrage under time-varying electricity prices,\\
%
%
$\quad \bullet$  \textit{Cycles of operation}: In our prior work \cite{hashmi2018long} we develop a mechanism to measure the number of cycles of operation based on depth-of-discharge (DoD) of energy storage operational cycles. Equivalent cycles of 100\% DoD are identified. This index provides information about how much the battery is operated.

We use xC-yC notation to represent the relationship between ramp rate and battery capacity. xC-yC implies battery takes 1/x hours to charge and 1/y hours to discharge completely.
We perform sensitivity analysis with (a) four battery models with the different ramping capability listed in Table~\ref{parametersBatlp} and (b) 5 levels of the ratio of selling price and buying price of electricity, i.e., $\kappa\in\{1, 0.75, 0.5, 0.25, 0\}$. 
	\begin{table}[!htbp]
		\small
		\caption {Battery Parameters}
		\label{parametersBatlp}
		\vspace{-8pt}
		\begin{center}
			\begin{tabular}{| c | c|}
				\hline
				$b_{\min}, b_{\max}, b_{0}$& 200Wh, 2000 Wh, 1000 Wh\\
				\hline
				$\eta_{\text{ch}}=\eta_{\text{dis}}$ & 0.95\\
				\hline
				$\delta_{\max} = - \delta_{\min}$ & 500 W for 0.25C-0.25C,\\ (4 battery model)&1000 W for 0.5C-0.5C\\
				& 2000 W for 1C-1C, \\& 4000 W for 2C-2C\\
				\hline
			\end{tabular}
			\hfill\
		\end{center}
	\end{table}
The optimization problem, $\text{P}_{\text{LP}}$, is solved using \texttt{linprog} in MATLAB\footnote{{https://www.mathworks.com/help/optim/ug/linprog.html}}. 
\texttt{linprog} uses dual-simplex \cite{andersen1995presolving} (default) algorithm.
\subsection{Deterministic Simulations}
The price data for our simulations in this subsection is taken from NYISO \cite{nyiso}. The load and generation data is taken from data collected at Madeira, Portugal.
Fig.~\ref{determinCase} shows the electricity price and energy consumption (includes inelastic load and rooftop solar generation) data used for deterministic simulations.
\begin{figure}[!htbp]
	\center
	\includegraphics[width=3.3in]{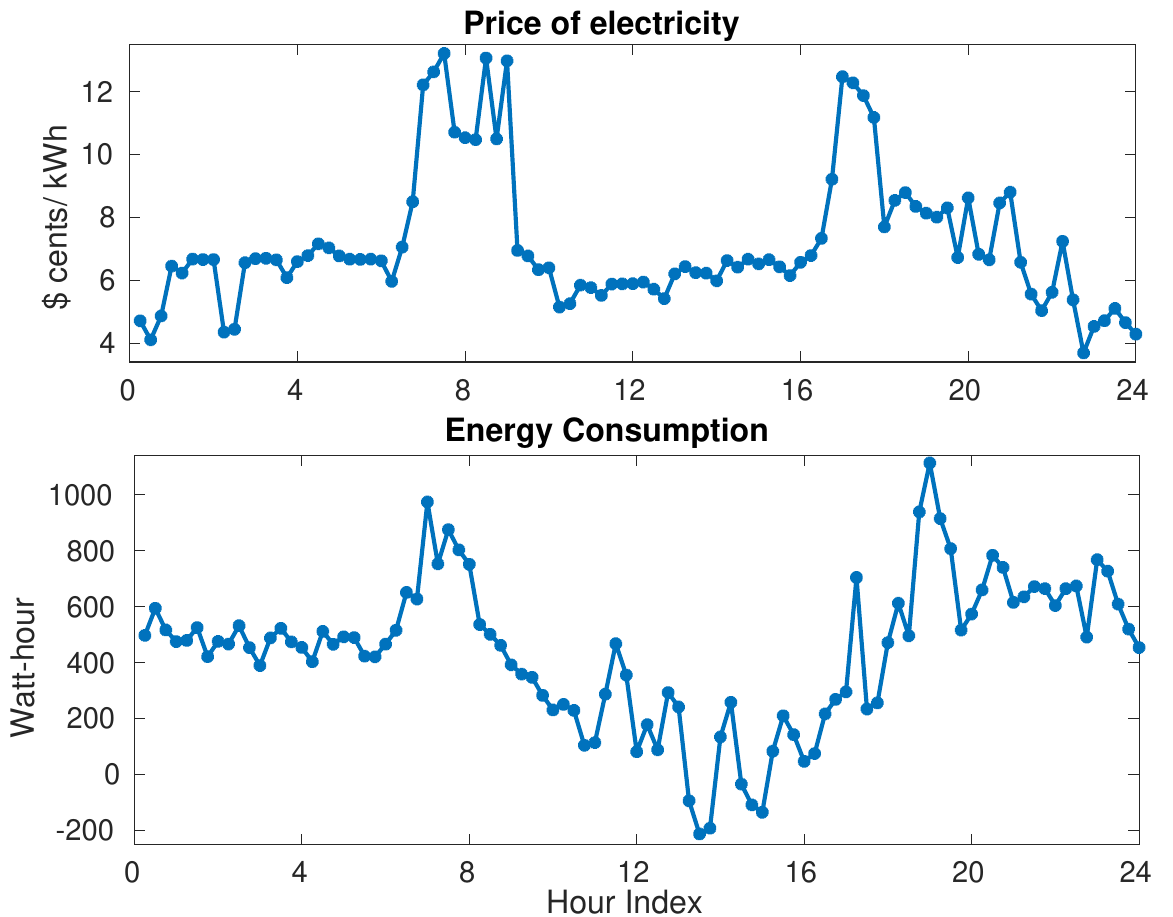}
	\vspace{-5pt}
	\caption{\small{Electricity price and consumer net load data used for deterministic simulations.}}\label{determinCase}
\end{figure}
Table~\ref{resultOnlydeterminLP} and Table~\ref{resultLoaddeterminLP} lists the energy storage arbitrage without and with energy consumption load for the electricity price data shown in Fig.~\ref{determinCase}.
The observations are:\\ 
$\quad \bullet$ The value of storage in presence of load and renewable increases as $\kappa$ decreases. Note that for $\kappa=0$, the only storage operation provides zero gain (see Table~\ref{resultOnlydeterminLP}), however, for the same buying and selling levels, the consumer would make significant gains when operated with inelastic load and renewable generation (see Table~\ref{resultLoaddeterminLP}),\\ 
$\quad \bullet$ The cycles of operation for faster ramping batteries are higher compared to slower ramping batteries. This implies that faster ramping batteries should be compared in terms of gains per cycle with slower ramping batteries. Observing only gains could be misleading.\\ 
$\quad \bullet$ As $\kappa$ decreases, the cycles of operation decrease, thus the effect on storage operation in the cases presented is similar to $\eta_{\text{fric}}$ in reducing cycles of operation.\\ 
$\quad \bullet$ Note that for $\kappa=1$, the arbitrage gains with and without load are the same. This observation is in sync with claims made in \cite{hashmi2017optimal}. Authors in \cite{hashmi2017optimal} observe that storage operation becomes independent of load and renewable variation for equal buying and selling case.
\begin{table}[!htbp]
	\small
	\caption {Performance indices for only storage}
	\label{resultOnlydeterminLP}
	\vspace{-7pt}
	\begin{center}
		\begin{tabular}{| c| c| c|c| c| }
			\hline
			$\kappa$ & 2C-2C & 1C-1C & 0.5C-0.5C & 0.25C-0.25C \\
			\hline
			\hline
			\multicolumn{5}{|c|}{Arbitrage gains in \$ cents  for 1 day} \\
			\hline
			1 & 44.445 & 33.760 & 25.636 & 17.536 \\
			0.75 & 18.842 & 17.668 & 14.077 & 9.921 \\
			0.5 & 7.682 & 7.088 & 6.253 & 5.219 \\
			0.25 & 2.513 & 2.502 & 2.483 & 2.422 \\
			0 & 0 & 0 & 0 & 0 \\
			\hline
			\multicolumn{5}{|c|}{Cycles of operation  for 1 day} \\		
			\hline
			1 & 6.586 & 3.856  & 2.237  & 1.620 \\
			0.75 & 2.401  & 1.742  &  1.484  &  0.795 \\
			0.5 &  1.539 & 1.099  & 0.714  &  0.386 \\
			0.25 &  0.182  & 0.171   & 0.164   & 0.160  \\
			0 & 0 & 0 & 0 & 0 \\
			\hline
		\end{tabular}
		\hfill\
	\end{center}
\end{table}
\begin{table}[!htbp]
	\small
	\caption {Performance indices for storage + load}
	\label{resultLoaddeterminLP}
	\vspace{-7pt}
	\begin{center}
		\begin{tabular}{| c| c| c|c| c| }
			\hline
			$\kappa$ & 2C-2C & 1C-1C & 0.5C-0.5C & 0.25C-0.25C \\
			\hline
			\hline
			\multicolumn{5}{|c|}{Arbitrage gains in \$ cents  for 1 day} \\
			\hline
			1 & 44.445 & 33.760 & 25.636 & 17.536 \\
			0.75 & 37.848 & 33.023 & 26.469 & 18.337 \\
			0.5 & 39.045 & 34.105 & 27.696 & 19.344 \\
			0.25 & 40.272 & 35.332 & 28.923 & 20.351 \\
			0 & 41.500 & 36.560 & 30.150 & 21.358 \\
			\hline
			\multicolumn{5}{|c|}{Cycles of operation  for 1 day} \\		
			\hline
			1 & 6.586 &  3.835 & 2.263 &  1.620 \\
			0.75 & 5.986  & 4.039  & 2.338 & 1.652  \\
			0.5 &  5.986  & 4.033  & 2.364 &  1.660 \\
			0.25 & 5.986  & 4.033  & 2.364  & 1.660  \\
			0 & 5.986  & 4.033  &  2.364 &  1.660 \\
			\hline
		\end{tabular}
		\hfill\
	\end{center}
\end{table}
\begin{figure}[!htbp]
	\center
	\includegraphics[width=3.3in]{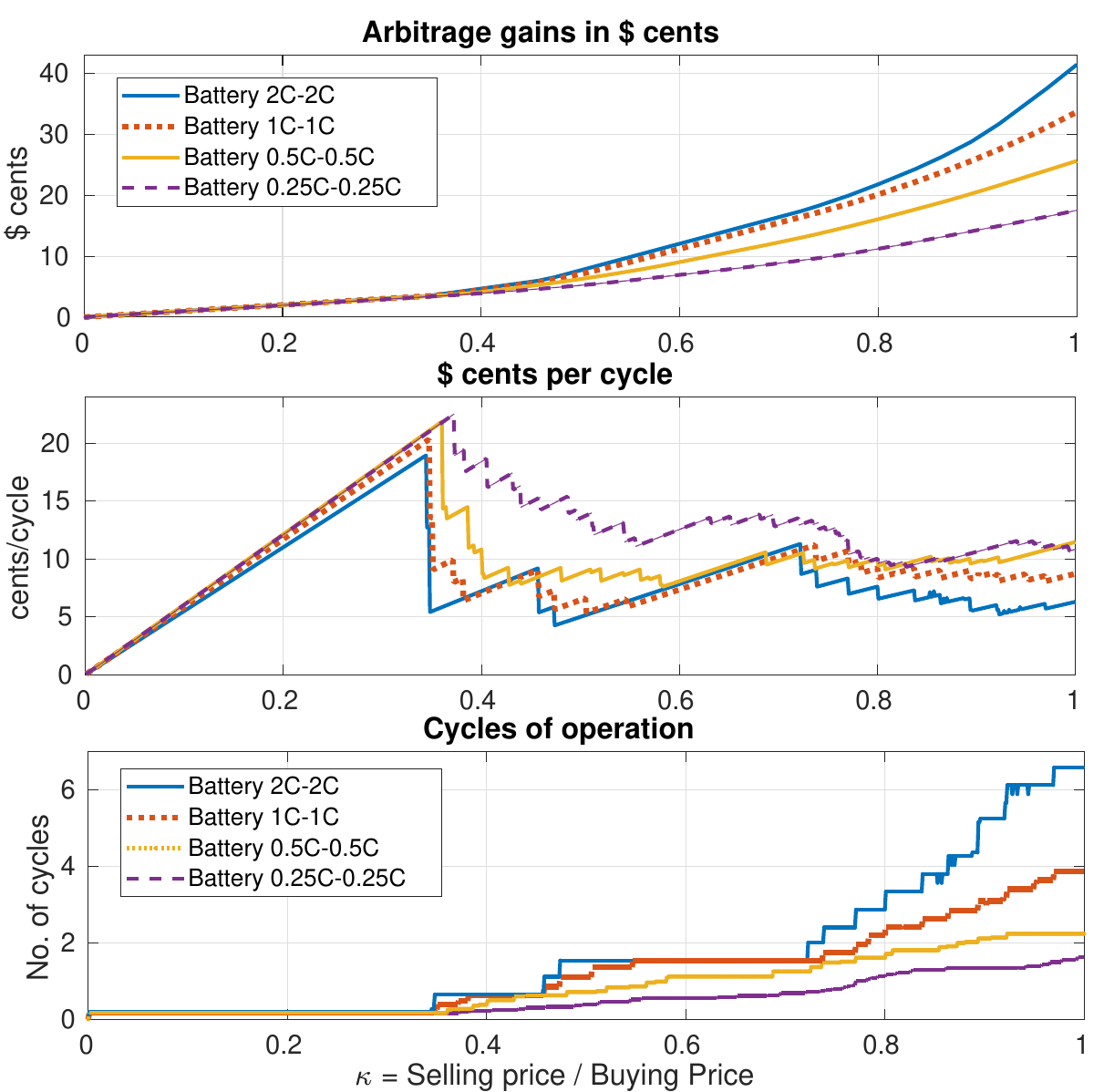}
	\vspace{-5pt}
	\caption{\small{Performance indices for only storage performing arbitrage with varying $\kappa$ for 1 day.}}\label{onlyStore}
\end{figure}
\begin{figure}[!htbp]
	\center
	\includegraphics[width=3.3in]{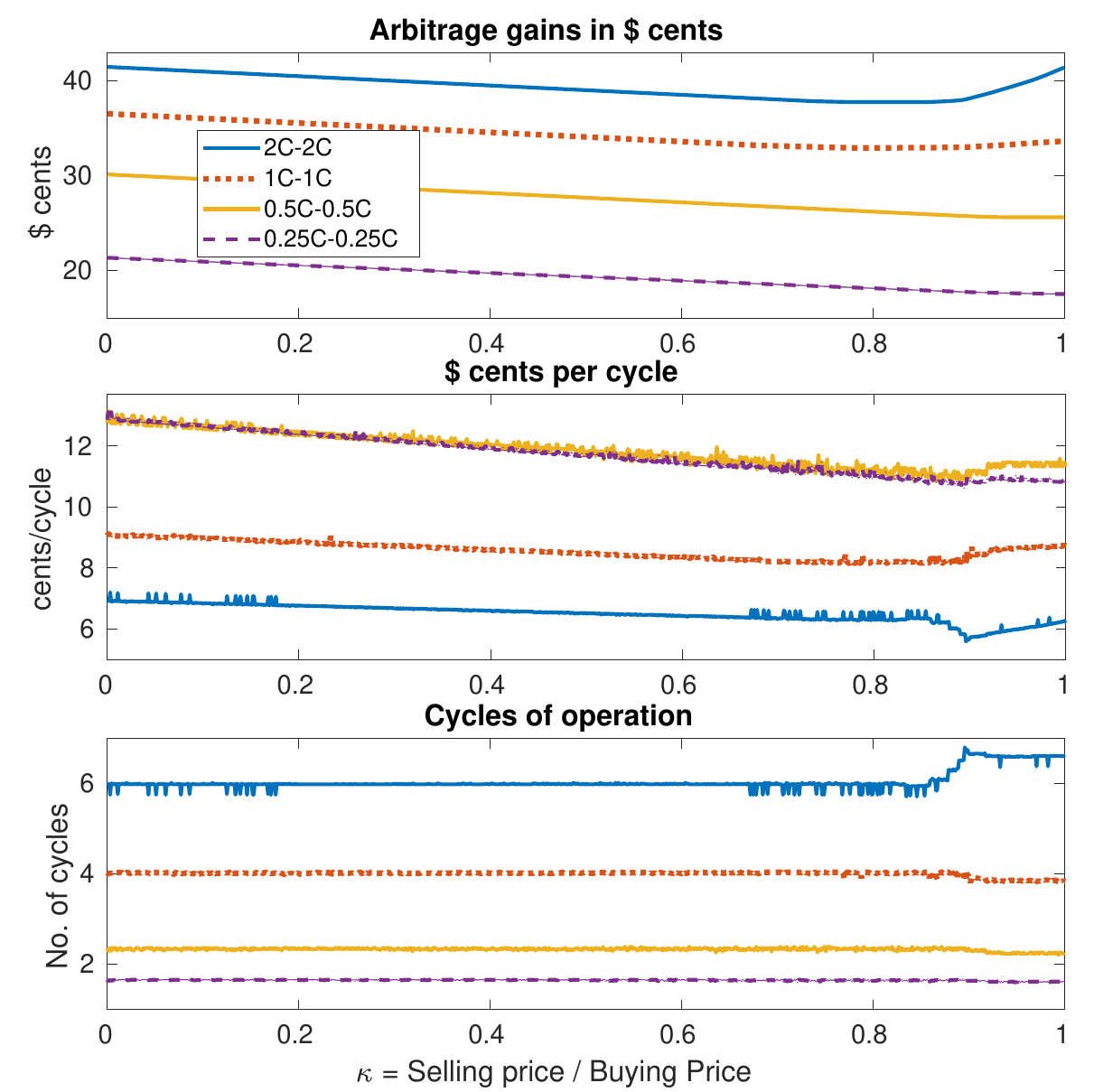}
	\vspace{-5pt}
	\caption{\small{Storage along with inelastic load and renewable generation with varying $\kappa$ for 1 day.}}\label{loadStore}
\end{figure}

Fig.~\ref{onlyStore} and Fig.~\ref{loadStore} show the arbitrage gains, gains per cycle and cycles of operation with varying $\kappa$ for storage performing arbitrage without and with inelastic load and renewable generation. The gains per cycle are nearly flat with varying $\kappa$. Slow ramping batteries, 0.25C-0.25C and 0.5C-0.5C, have significantly higher gains per cycle compared to faster ramping batteries, 1C-1C and 2C-2C.
\subsection{Results with Uncertainty}	
The forecast model is generated for load with solar generation and for electricity price.
The ARMA based forecast uses 9 weeks of data (starting from 29th May, 2019) for training and generates forecast for the next week.
\texttt{ForecastMPClinearProgram} is implemented in receding horizon.
The electricity price data used for this numerical experiment is taken from CAISO \cite{enonlinecalifornia} for the same days of load data. 
To compare the effect of forecasting net load and electricity prices with perfect information, we present average arbitrage gains and cycles of operation starting from 1st June 2019. Rolling horizon time-period of optimization, $N_{\text{opt}}$, is selected as 1 day. This implies at 13:00 h today, the storage control decisions are based on parameter variation forecasts till 13:00 h tomorrow.
\begin{table}[!htbp]
	\small
	\caption {Deterministic arbitrage gains for only storage}
	\label{deterministictOnlyLP}
	\vspace{-7pt}
	\begin{center}
		\begin{tabular}{| c| c| c|c| c| }
			\hline
			$\kappa$ & 2C-2C & 1C-1C & 0.5C-0.5C & 0.25C-0.25C \\
			\hline
			\hline
			\multicolumn{5}{|c|}{Arbitrage gains in \$ for 1 week} \\
			\hline
			1 & 9.411 & 7.059 & 4.784 & 3.065 \\
			0.75 &  5.729 & 4.491  & 3.168  & 2.082 \\
			0.5 & 3.166 & 2.550 & 1.833 & 1.217 \\
			0.25 & 1.124 &  0.941 &  0.688 & 0.456  \\
			0 & 0 & 0 & 0 & 0 \\
			\hline
			\multicolumn{5}{|c|}{Cycles of operation for 1 week} \\		
			\hline
			1 & 58.729 & 37.257 & 21.324 & 12.107 \\
			0.75 &  23.462 & 16.341  & 10.746  & 7.519 \\
			0.5 & 12.689 & 9.770 & 7.579 & 6.174 \\
			0.25 & 7.727 &  6.229 &  4.558 & 3.464  \\
			0 & 0 & 0 & 0 & 0 \\
			\hline
		\end{tabular}
		\hfill\
	\end{center}
\end{table}

\begin{table}[!htbp]
	\small
	\caption {Deterministic arbitrage gains for storage with load}
	\label{deterministictLoadLP}
	\vspace{-7pt}
	\begin{center}
		\begin{tabular}{| c| c| c|c| c| }
			\hline
			$\kappa$ & 2C-2C & 1C-1C & 0.5C-0.5C & 0.25C-0.25C \\
			\hline
			\hline
			\multicolumn{5}{|c|}{Arbitrage gains in \$ for 1 week} \\
			\hline
			1 & 9.411 & 7.059 & 4.784 & 3.065 \\
			0.75 &  7.462  & 6.269  & 4.540  & 3.025 \\
			0.5 & 6.641 & 5.987 & 4.468 &  3.019 \\
			0.25 & 6.350 & 5.904 & 4.451  &  3.019 \\
			0 & 6.313 & 5.902 & 4.451 &  3.019 \\
			\hline
			\multicolumn{5}{|c|}{Cycles of operation for 1 week} \\		
			\hline
			1 & 58.700 & 37.294 & 21.324 & 12.107 \\
			0.75 & 28.583  & 20.809  & 14.382  & 10.229 \\
			0.5 &  19.296 & 16.629 & 13.007 & 9.971  \\
			0.25 & 16.591 & 15.348 &  12.498 & 9.968   \\
			0 & 16.041 & 15.201 & 12.484 & 9.968  \\
			\hline
		\end{tabular}
		\hfill\
	\end{center}
\end{table}

\begin{table}[!htbp]
	\small
	\caption {Real-time implementation for only storage}
	\label{StochasticOnlyLP2}
	\vspace{-7pt}
	\begin{center}
		\begin{tabular}{| c| c| c|c| c| }
			\hline
			$\kappa$ & 2C-2C & 1C-1C & 0.5C-0.5C & 0.25C-0.25C \\
			\hline
			\hline
			\multicolumn{5}{|c|}{Arbitrage gains in \$ for 1 week} \\
			\hline
			1 &  6.035 & 4.684  & 3.469 & 3.000\\
			0.75 & 5.024 & 4.118 & 3.081  & 1.904\\
			0.5 & 3.004 & 2.367 & 1.692 &  1.110 \\
			0.25 & 1.067 & 0.891 &  0.618 &  0.442 \\
			\hline
			\multicolumn{5}{|c|}{Cycles of operation for 1 week} \\
			\hline
			1 & 64.323  & 38.979  & 22.622  & 12.850  \\
			0.75 & 24.870  &  16.169  & 10.570  &  7.733 \\
			0.5 & 11.393 & 8.891 & 7.013 & 6.099 \\
			0.25 & 6.429 & 5.557 & 4.359 &  3.395  \\
			\hline
		\end{tabular}
		\hfill\
	\end{center}
\end{table}

The deterministic results for without and with load are presented in Table~\ref{deterministictOnlyLP} and Table~\ref{deterministictLoadLP}. Compare the deterministic results with stochastic results presented in Table~\ref{StochasticOnlyLP2} and Table~\ref{StochasticLoadLP2}.
The primary numerical observations are:\\ 
$\quad \bullet$ Effect of uncertainty on arbitrage gains for a faster ramping battery is greater compared to a slower ramping battery, this observation is in sync with conclusions drawn in \cite{yize2018stochastic}.\\ 
$\quad \bullet$ Combining storage with inelastic load with renewable generation provides greater gains for decreasing $\kappa$. Furthermore, the effect of uncertainty for lower $\kappa$ is lower compared to higher values of $\kappa$.\\ 
$\quad \bullet$ Profitability of operating only storage deteriorates sharply with decrease of $\kappa$. For only storage case under zero selling price case ($\kappa=0$) no arbitrage would be possible and the gain remains zero.

\begin{table}[!htbp]
	\small
	\caption {Real-time implementation for storage with load}
	\label{StochasticLoadLP2}
	\vspace{-7pt}
	\begin{center}
		\begin{tabular}{| c| c| c|c| c| }
			\hline
			$\kappa$ & 2C-2C & 1C-1C & 0.5C-0.5C & 0.25C-0.25C \\
			\hline
			\hline
			\multicolumn{5}{|c|}{Arbitrage gains in \$ for 1 week} \\
			\hline
			1 &  6.034 &  4.684  & 3.496  &  3.000 \\
			0.75 & 4.827   &  4.075  &  3.400 & 2.987  \\
			0.5 & 4.168  &  3.711 & 3.292  & 2.975  \\
			0.25 &  4.204 &  3.943 &  3.348  &  3.002  \\
			0 & 4.427  & 3.896  &  3.396 &  3.009 \\
			\hline
			\multicolumn{5}{|c|}{Cycles of operation for 1 week} \\
			\hline
			1 & 64.322  & 38.979  &  22.622 &  12.850 \\
			0.75 &  41.613  & 30.322   & 19.948  & 11.980  \\
			0.5 &  34.658 & 27.627  &  18.744 &  11.348 \\
			0.25 &  31.429 &  26.370 &  18.476  &  11.396  \\
			0 & 32.958  & 28.255  & 19.845  & 11.372  \\
			\hline
		\end{tabular}
		\hfill\
	\end{center}
\end{table}
\section{Conclusion}	
\label{lpsec6}
We formulate energy storage arbitrage problem using linear programming. 
The linear programming formulation is possible due to piecewise linear convex cost functions.
In this formulation we consider: (a) net-metering compensation (with selling price at best equal to buying price) i.e. $\kappa_i \in [0,1]$, (b) inelastic load, (c) consumer renewable generation, (d) storage charging and discharging losses, (e) storage ramping constraint and (f) storage capacity constraint. 
%
By conducting extensive numerical simulations, we analyze the sensitivity of energy storage for varying ramp rates and varying ratio of selling and buying price of electricity. We observe that the value of storage in presence of load and renewable increases as the ratio of selling and buying price decreases.
We also perform real-time implementation of the proposed LP formulation and compare the deterministic results with net-load and electricity price uncertainties.
Net-load and electricity price are modeled with AutoRegressive models for model predictive control. The effect of uncertainty on slow ramping batteries is observed to be lower compared to faster ramping batteries. Furthermore, as $\kappa$ decreases, arbitrage gains becomes less sensitive to uncertainty.

In a future work, we aim to control the cycles of operation of the battery by tuning the friction coefficient with different $\kappa$ values, such that the battery is not over-used, otherwise this would lead to reduction in battery operational life.
\bibliographystyle{IEEEtran}
\bibliography{bibLP}

\vspace{10pt}

\appendix
\section{Epigraph formulation of Linear Programming}
\label{epigraphsec}
\begin{center}
	{Epigraph formulation of Linear Programming}
\end{center}

An unconstrained minimization problem of a convex piecewise-linear function, $h(x)$, could be transformed to an equivalent linear programming problem by forming the epigraph problem \cite{boyd2004convex}, \cite{piecewise}. 
Consider the convex piecewise cost function minimization problem is denoted as
$
(P_{org}) \quad \min h(x),
$
where $h(x) = \max_{i=1,...,m}(a_i^Tx + b_i)$. For cases where the decision variable $x$ is scaler, $a_i^T$ is also a scaler. Thus, $a_ix + b_i$ is a two-dimensional line with $b_i$ denoting the y-intercept and $a_i$ the slope of the line. 
The equivalent epigraph problem for the original problem $P_{org}$ is denoted as
$
(P_{epi}) \quad \min t,~
\text{ subject to,  } a_ix + b_i \leq t, \quad i =1,...,m,
$~
where $t$ denotes auxiliary scalar variable.
The LP matrix notation for the optimization problem $P_{epi}$ is represented as: minimize
$\tilde{f}^T \tilde{x}$, subject to $\tilde{A}\tilde{x}\leq \tilde{b}$; where

\begin{gather*}
\tilde{f}={\begin{bmatrix}
	0\\
	1\\
	\end{bmatrix}},
\quad
\tilde{x} = {\begin{bmatrix}
	x\\
	t\\
	\end{bmatrix}},
\quad
\tilde{A} = {\begin{bmatrix}
	a_1 & -1\\
	: & :\\
	a_m & -1\\
	\end{bmatrix}},
\quad
\tilde{b} = {\begin{bmatrix}
	-b_1 \\
	: \\
	-b_m \\
	\end{bmatrix}}.
\end{gather*}

Now consider extending this minimization problem for two time instants with a unique cost function for each time instant. The optimization problem is denoted as
$
(P_{epi}) \quad \min \quad t_1+t_2, 
\text{ s.t.,  (i) } a_{1i}x + b_{1i} \leq t_1,
\text{ (ii) } a_{2i}x + b_{2i} \leq t_2,
~ i =1,...,m,
$
The equivalent LP matrices are denoted as

\begin{gather*}
\tilde{f}\text{=}{\begin{bmatrix}
	0\\
	0\\
	1\\
	1\\
	\end{bmatrix}},
\tilde{x} \text{=} {\begin{bmatrix}
	x_1\\
	x_2\\
	t_1\\
	t_2\\
	\end{bmatrix}},
\tilde{A} \text{=} {\begin{bmatrix}
	a_{11} & 0 &-1 & 0\\
	: & : & : & :\\
	a_{1m} & 0 &-1 & 0\\
	0 & a_{21} & 0 &-1 \\
	: & : & : & :\\
	0 &a_{1m} & 0 &-1\\
	\end{bmatrix}},
\tilde{b} \text{=} {\begin{bmatrix}
	-b_{11} \\
	: \\
	-b_{1m} \\
	-b_{21} \\
	: \\
	-b_{2m} \\
	\end{bmatrix}}.
\end{gather*}

A similar LP formulation for N time steps with piecewise linear cost function can be formulated.

\section*{Acknowledgement}
The numerical results use the Madeira electricity consumer data 
collected under the framework of the H2020 SMILE project (GA 731249).
We would like to thank Dr Lucas Pereira for providing the data.	
	
\end{document}